\documentclass[a4paper,11pt]{article}
\pdfoutput=1

\usepackage{algorithm,algorithmic}
\usepackage{amsmath,amssymb}
\usepackage{graphicx}
\usepackage[unicode]{hyperref}
\usepackage{fullpage}
\usepackage{mathtools}
\usepackage{amsthm}
\usepackage{engord}
\usepackage{bookmark}
\usepackage{flafter}
\usepackage[nodayofweek]{datetime}
\usepackage{eucal}
\usepackage{dsfont}
\usepackage[T1]{fontenc}
\usepackage{lmodern}
\usepackage{authblk}
\usepackage{enumitem}
\usepackage{subcaption}
\usepackage[stretch=10,shrink=10]{microtype}
\usepackage[capitalise]{cleveref}
\usepackage{tikz}

\newcommand {\br} [1] {\ensuremath{ \left( #1 \right) }}
\newcommand {\Br} [1] {\ensuremath{ \left[ #1 \right] }}
\newcommand {\cbr} [1] {\ensuremath{ \left\lbrace #1 \right\rbrace }}
\newcommand {\minusspace} {\: \! \!}

\newcommand {\fn} [2] {\ensuremath{ #1 \minusspace \br{ #2 } }}
\newcommand {\Fn} [2] {\ensuremath{ #1 \minusspace \Br{ #2 } }}
\newcommand {\eqdef} {\ensuremath{ \stackrel{\mathrm{def}}{=} }}
\newcommand {\Prob} [1] {\Fn{\Pr}{#1}}

\newcommand {\bigo} [1] {\fn{O}{#1}}

\newcommand {\ket} [1] {\ensuremath{ \left| #1 \right\rangle }}
\newcommand {\ketbratwo} [2] {\ensuremath{ \left| #1 \middle\rangle \middle\langle #2 \right| }}
\newcommand {\ketbra} [1] {\ketbratwo{#1}{#1}}

\newcommand {\ceil} [1] {\ensuremath{ \left\lceil #1 \right\rceil }}

\newcommand{\register}[1]{\ensuremath{\mathsf{#1}}}

\newcommand{\inR}{\ensuremath{\in_{R}}}
\newcommand{\norm}[1]{\ensuremath{ \left\| #1 \right\| }}

\algsetup{indent=2em}

\newcommand{\email}[1]{\newline\texttt{#1}}
\newcommand{\keywords}[1]{\paragraph{Keywords:} #1}

\theoremstyle{plain}
\newtheorem{theorem}{Theorem}[section]
\theoremstyle{remark}
\newtheorem{remark}[theorem]{Remark}

\newcommand{\papertitle}{Quantum Inspired Adaptive Boosting}
\hypersetup{
	pdfstartview={FitH},
	pdfdisplaydoctitle={true},
	breaklinks={true},
	colorlinks={true},
	linkcolor={black},
	citecolor={black},
	urlcolor={black},
	bookmarksopen={true},
	bookmarksnumbered={false},
	pdftitle={\papertitle}
}
\crefformat{footnote}{#2\footnotemark[#1]#3}
\date{\formatdate{1}{2}{2021}}

\begin{document}

\title{\textbf{\papertitle}\thanks{%
BD and AP was supported by the Ministry of Innovation and the National Research,
	Development and Innovation Office within the framework of the
	Artificial Intelligence National Laboratory Program.
BD was also supported by MTA Premium Postdoctoral Grant 2018.
The research  of KF, LK, DSz  has been supported by the National
Research, Development and Innovation Fund (TUDFO/51757/2019-ITM,
Thematic Excellence Program), by the BME-Artificial Intelligence FIKP
grant of EMMI (BME FIKP-MI/SC), and by the BME NC TKP2020 grant of
NKFIH Hungary.
DSz was also supported by UNKP-18-1 New National Excellence Program of
the Ministry of Human Capacities.
AP was also supported by the European Research Council (ERC) Consolidator Grant SYSTEMATICGRAPH (no. 725978).}}

\author[1]{B\'alint Dar\'oczy}
\affil[1]{Institute for Computer Science and Control (SZTAKI),
	E\"otv\"os Lor\'and Research Network (ELKH),
	Budapest, Hungary\email{\{daroczyb,peresz\}@ilab.sztaki.hu}}
\author[2]{Katalin Friedl}
\affil[2]{Department of Computer Science and Information Theory,
	Budapest University of Technology and Economics,
	Budapest, Hungary\email{\{friedl,kabodil,szabod\}@cs.bme.hu}}
\author[2]{L\'aszl\'o Kab\'odi}
\author[1]{Attila Pereszl\'enyi}
\author[2]{D\'aniel Szab\'o}

\maketitle

\begin{abstract}
Building on the quantum ensemble based classifier algorithm of
Schuld and Petruccione \cite{schuld2018quantum}, we devise equivalent
classical algorithms which show that this quantum ensemble method does not
have advantage over classical algorithms.
Essentially, we simplify their algorithm until it is intuitive to come
up with an equivalent classical version.
One of the classical algorithms is
extremely simple and runs in constant time for each input to be classified.
We further develop the idea and, as the main contribution of the paper,
we propose methods inspired by combining the quantum ensemble
method with adaptive boosting.
The algorithms were tested and found to be comparable to the AdaBoost
algorithm on publicly available data sets.
\end{abstract}


\section{Introduction}

Recent results of quantum machine learning models \cite{schuld2015introduction,biamonte2017quantum,alvarez2017supervised,Kerenidis2018} mostly consider gradient descent optimization \cite{gilyen2019optimizing} or suggest quantum circuits for classification models, e.g., support vector machines \cite{rebentrost2014quantum}, neural networks \cite{Allcock2018} with tensor networks \cite{huggins2018towards}, quantum ensembles \cite{schuld2018quantum}.

In the classical setting,
bagging and boosting (or ensemble) methods are among the most efficient meta machine learning algorithms.
Bootstrap aggregating (or bagging in short) increases generalization via resampling the available training data during the algorithm. In comparison, boosting determines a set of ``weak'' performing models in an iterative process.
In every step the algorithm selects a model which improves the already selected models either on instance (sample point) or feature (attribute) levels.
Since complex models are capable of learning the training set almost perfectly, in a typical setting, the learners are less capable or ``weak'' otherwise the algorithm may stop after the first iteration.
Besides Adaptive Boosting (or AdaBoost, proposed by Freund and Schapire in \cite{freund1997decision}) Gradient Boosting \cite{friedman2001greedy} and Random Forest \cite{breiman2001random} make decisions as a linear combination over a finite set of ``weak'' classifiers. In comparison to AdaBoost, the quantum ensemble in \cite{schuld2018quantum} can be interpreted as a stochastic boosting algorithm on a finite set of learners where the learners are previously determined.

\subsection{Our contribution}
Recent results on quantum machine learning show that many proposed
quantum algorithms do not have advantage over classical algorithms
\cite{Tang2019,tang2018quantum,Chia2020}.
We investigated the quantum classifier of Schuld and Petruccione
\cite{schuld2018quantum} that selects a weak learner according
to a distribution based on how good they are.
We simplified their algorithm to the point where it is intuitively easy
to give an equivalent classical algorithm.
We show that a simple classical randomized method achieves the same result without changing the time complexity.
After that, an even simpler, constant time classical method is given.

We note that, independently to our work, Abbas, Schuld, and Petruccione \cite{abbas2020quantum}
also shown that the ensemble method can be turned into a classical algorithm.
Our construction, however, is arguably simpler and more direct, especially
the constant time algorithm.

Our main contribution is a classical algorithm that combines these ideas with boosting and presents a method where the classification is based on a weighted average of the weak learners.
Two types of weights are involved that are computed alternately:
one on the samples, representing how difficult they are to learn
and one on the learners, representing how well they perform on the chosen samples.
Here we consider only the case of binary classification but the method can be extended.
In our experiments we have different implementations of the described algorithms.
A preliminary version of this paper has appeared in \cite{QuantumBoostingPrelim}.

\subsection{Organization of the paper}
For completeness, in \cref{sec:original classifier} we describe the idea of Schuld and Petruccione \cite{schuld2018quantum}.
In \cref{sec:classical equivalent}, we show our simple classical method that achieves the same result.
After that, an even simpler, constant time classical method is given in \cref{sec:constant time}.
\Cref{sec:adaptive stochastic} is the main part where we present
our algorithm that combines the ensemble method with adaptive boosting.
Our experiments with these algorithms are summarized in \cref{sec:comparisons}.

\section{Problem description}
Let \Br{N} denote the set \cbr{1, \ldots, N}.
Given $N$ training samples $\underline{x}_i\in \mathbb{R}^d$,
for each $i \in \Br{N}$ let $y_i \in \cbr{0,1}$ denote the label of $\underline{x}_i$.
A (binary) classifier is any function $\mathbb{R}^d \rightarrow \cbr{0,1}$.
Given $W$ of such classifiers $h_{\theta}$, $\theta \in \Br{W}$, which are called weak learners,
the main objective is to construct a classifier which approximates the $y_i$ well enough.
We assume that with each $h_{\theta}$ its negation is also among the given classifiers.

For a classifier $h$, let  $a_{i,h}=1-|h(\underline{x}_i)-y_i|$, i.e.,
$a_{i,h} =1$ when $h$ is correct on sample $\underline{x}_i$, otherwise $a_{i,h}=0$.
The accuracy of $h$ is the ratio of the number of samples where it is correct,
$a_h = \sum_{i=1}^N a_{i,h}/N$.
For the given classifiers $h_{\theta}$ we use the notation $a_{i,\theta}$ and $a_{\theta}$ instead of $a_{i,h_{\theta}}$ and $a_{h_{\theta}}$, respectively.

The main objective of classification is to construct a classifier $h$ which provides fairly good classification for unknown input. We evaluated the performance of the models with Area Under Curve (AUC) of the Receiver Operating Characteristics (ROC) \cite{green1966signal}.

\section{Quantum ensemble classifier}

Our work was inspired by the quantum classifier algorithm of Schuld
and Petruccione \cite{schuld2018quantum}. First, we give an outline of their algorithm.

\subsection{The original quantum ensemble classifier}
\label{sec:original classifier}

The algorithm uses five quantum registers. These are: \register{X} for the index of a training sample,
\register{Y} for the label of the training sample,
\register{H} for the index of a classifier,
\register{G} for the label the classifier assigns to the training sample,
and \register{Q} is an extra qubit.
This state requires $\ceil{\log N} + \ceil{\log W} + 3$ qubits.

For each $i \in \Br{N}$, we can
perform unitaries $U_i$ and $V_i$ for which
$U_i \ket{0 \ldots 0}_{\register{X}} = \ket{i}_{\register{X}}$ and
$V_i \ket{0}_{\register{Y}} = \ket{y_i}_{\register{Y}}$ where $y_i$ is the
label of $\underline{x}_i$.
These unitaries are efficiently constructible using elementary classical
reversible gates.
Furthermore, the evaluation of classifiers $h_j$ at $\underline{x_i}$
can also be performed by classical reversible gates.
\par

The initial state is the following:
\begin{align*}
	\ket{0 \ldots 0}_{\register{X}} \otimes
	\ket{0}_{\register{Y}} \otimes
	\ket{0 \ldots 0}_{\register{H}} \otimes
	\ket{0}_{\register{G}} \otimes
	\ket{0}_{\register{Q}}.
\end{align*}
Using Hadamard gates, a uniform superposition
on \register{H} and \register{Q} can be obtained:
\[	\frac{1}{\sqrt{2W}}
	\sum_{\theta = 1}^{W}
	\ket{0 \ldots 0}_{\register{X}} \otimes
	\ket{0}_{\register{Y}} \otimes
	\ket{\theta}_{\register{H}} \otimes
	\ket{0}_{\register{G}} \otimes
	\br{\ket{0}_{\register{Q}} + \ket{1}_{\register{Q}}}.
\]
We apply $U_1$ and $V_1$ which results in the following state
\[	\frac{1}{\sqrt{2W}}
	\sum_{\theta = 1}^{W}
	\ket{1}_{\register{X}} \otimes
	\ket{y_1}_{\register{Y}} \otimes
	\ket{\theta}_{\register{H}} \otimes
	\ket{0}_{\register{G}} \otimes
	\br{\ket{0}_{\register{Q}} + \ket{1}_{\register{Q}}}.
\]
The evaluation of classifiers $h_{\theta}$
can also be performed by classical reversible gates.
We flip qubit \register{G} if $\fn{h_{\theta}}{\underline{x}_1} = y_1$ (the classification is successful).
The resulting state is
\[	\frac{1}{\sqrt{2W}}
	\sum_{\theta = 1}^{W}
	\ket{1}_{\register{X}} \otimes
	\ket{y_1}_{\register{Y}} \otimes
	\ket{\theta}_{\register{H}} \otimes
	\ket{a_{1, \theta}}_{\register{G}} \otimes
	\br{\ket{0}_{\register{Q}} + \ket{1}_{\register{Q}}} .
\]
Now qubit \register{Q} is rotated around the $z$-axis on the Bloch sphere
by a small angle if $a_{1, \theta} = 1$, that is when the
algorithm successfully classified $\underline{x}_1$.
The operator corresponding to this rotation is
$\ketbra{0} + e^{j \pi / N} \ketbra{1}$
where $j$ is the imaginary unit.
The resulting state is
\begin{align*}
	\frac{1}{\sqrt{2W}}
	\sum_{\theta = 1}^{W}
	\ket{1}_{\register{X}} \otimes
	\ket{y_1}_{\register{Y}} \otimes
	\ket{\theta}_{\register{H}} \otimes
	\ket{a_{1, \theta}}_{\register{G}} \otimes
	\br{\ket{0}_{\register{Q}} + e^{j \pi a_{1, \theta} / N} \ket{1}_{\register{Q}}}.
\end{align*}
The next task is to uncompute everything except the previous rotation.
Then we repeat the above procedure with all the other training samples.
In the end the state is
\begin{align*}
	\frac{1}{\sqrt{2W}}
	\sum_{\theta = 1}^{W}
	\ket{0}_{\register{X}} \otimes
	\ket{0}_{\register{Y}} \otimes
	\ket{\theta}_{\register{H}} \otimes
	\ket{0}_{\register{G}} \otimes
	\br{\ket{0}_{\register{Q}} + e^{j \pi a_{\theta}} \ket{1}_{\register{Q}}} .
\end{align*}
After measuring \register{Q} in the \cbr{\ket{+}, \ket{-}} basis where
\begin{align*}
	\ket{+} \eqdef \frac{\ket{0} + \ket{1}}{\sqrt{2}}
	\qquad \text{and} \qquad
	\ket{-} \eqdef \frac{\ket{0} - \ket{1}}{\sqrt{2}} \text{,}
\end{align*}
we post-select on the measurement outcome being $\ket{-}_{\register{Q}}$.
This means that if we happen to project on $\ket{+}_{\register{Q}}$ then
we throw away the result and redo the whole algorithm from the beginning.
After this training phase we measure register \register{H} and use the resulting classifier
to classify an input with unknown label.

Here we deviate from the original algorithm because it makes the algorithm simpler.
In its original form, register \register{H} is not measured but kept in superposition.
In the end, however, only the one-qubit output register is measured and register \register{H}
is effectively traced out.
So measuring register \register{H} earlier does not change the output distribution
of the algorithm.
This can be verified by simple calculation.

\begin{theorem} \label{thm:quant}
For the above algorithm
\begin{align}
	\Prob{\text{Alg.\ returns } \theta}
	= \frac{1}{\chi} \fn{\sin^2}{\frac{\pi a_{\theta}}{2}}
	\qquad \text{where} \qquad
	\chi = \sum_{\vartheta=1}^{W} \fn{\sin^2}{\frac{\pi a_{\vartheta}}{2}}.
	\label{eqn:quant probability}
\end{align}
Assuming that the involved unitaries need constant time, the expected time complexity of the training phase of the algorithm is \bigo{N}.
The probability that an unknown input $\underline{x}$ will be classified as $1$ is
\begin{align*}
	 \frac{1}{\chi}
	\sum_{\theta : h_{\theta}(\underline{x}) = 1}
	\fn{\sin^2}{\frac{\pi a_{\theta}}{2}}.
\end{align*}
\end{theorem}

\begin{remark}
	The formula for the probability in \cref{eqn:quant probability} is different
	from what is claimed in \cite{schuld2018quantum} where it is said that the probability
	is proportional to $a_{\theta}$.
	We believe this is an error in the cited paper.
	(This error is also present in \cite{abbas2020quantum}.)
\end{remark}

\begin{proof}[Proof of \cref{thm:quant}]
The claim about the probability to obtain $\theta$ follows immediately.

For the time complexity we use the assumption that together with a classifier $h_{\theta}$ its negation is also present.
This ensures that the post-selection fails with probability at most $3/4$, so the expected number of repetitions is constant and the total time is \bigo{N}.

For the last claim, notice that
the result of the training phase is a random variable $\Theta$ for which
\begin{align*}
	\Prob{h_{\Theta} \text{ classifies $\underline{x}$ as } 1}
	= \frac{1}{\chi}
	\sum_{\vartheta : h_{\vartheta}(\underline{x}) = 1}
	\fn{\sin^2}{\frac{\pi a_{\vartheta}}{2}}. \tag*{\qedhere}
\end{align*}
\end{proof}

\subsection{A classical equivalent}
\label{sec:classical equivalent}

First, we note that registers \register{X} and \register{Y} are unnecessary
to include in the ``quantum part'' of the algorithm
as they are always in a classical state.
Essentially, what the algorithm does is it rotates a qubit
depending on the training data and the classifier $\theta$
and after that the qubit is measured.
This gives a probability distribution on $\theta$.
This simplicity suggests that the same distribution on $\theta$
can be obtained by a classical randomized procedure.
This algorithm is presented in \cref{alg:classic_o_n}.

\begin{algorithm}
\caption{Classical version} \label{alg:classic_o_n}
\raggedright
Let $g: \Br{0,1} \to \Br{0,1}$ be an ``easily computable'',
monotone increasing function.\\[-2ex]
\begin{algorithmic}[1]
	\STATE \label{step:beginning}
	Pick $\theta \inR \Br{W}$ uniformly at random.
	\STATE For all $i \in \Br{N}$ calculate $a_{i,\theta}$ and from them calculate
	$a_{\theta}$.
	\STATE Pick $r \inR \Br{0,1}$ uniformly and independently from $\theta$ at random.
	\STATE If $r > \fn{g}{a_{\theta}}$ then go back to step~\ref{step:beginning}, otherwise output $\theta$.
\end{algorithmic}
\end{algorithm}

We prove that this algorithm with the right choice of $g$ is as good as the quantum algorithm of \cref{sec:original classifier}.
\begin{theorem}
	The output of \cref{alg:classic_o_n} is
\begin{align}
	\Prob{\text{Alg.\ returns } \theta}
	= \frac{\fn{g}{a_{\theta}}}{\chi}
	\qquad \text{where} \qquad
	\chi =\sum_{\vartheta = 1}^{W} \fn{g}{a_{\vartheta}}
	\label{eqn:general probability}
\end{align}
Assuming that the computation of $\fn{g}{a}$ takes constant time, the expected time is \bigo{N}.
The probability that an unknown input $\underline{x}$ will be classified as $1$ is
\begin{align*}
	 \Prob{h_{\Theta} \text{ classifies $\underline{x}$ as } 1}
	=\frac{1}{\chi}
	\sum_{\theta : h_{\theta}(\underline{x}) = 1}
	\fn{g}{a_{\theta}}.
\end{align*}
\end{theorem}

\begin{proof}
The probability that in one round the output is $\theta$, equals to
\[ \frac{\fn{g}{a_{\theta}}}{W} \text{,} \]
which implies the claim about the probability.
The rest is basically the same as for the quantum version.
\end{proof}

\begin{remark}
Notice that in the case $\fn{g}{a} = \fn{\sin^2}{\frac{\pi a}{2}}$ the results are the same as for the quantum algorithm of \cref{sec:original classifier}.
\end{remark}

\subsection{A constant time classical classifier}
\label{sec:constant time}

The previous classical algorithm,
for the case originally claimed in \cite{schuld2018quantum}, i.e.,
when the probabilities are proportional to the accuracies,
can be simplified to achieve constant expected running time.
\Cref{alg:classic_o_1} presents this very simple classical algorithm.

\begin{algorithm}[htb]
\caption{Classical constant time version}
\label{alg:classic_o_1}
\begin{algorithmic}[1]
	\STATE Pick $i \inR \Br{N}$ and $\theta \inR \Br{W}$ uniformly and independently
	at random.
	\label{step:simplest first}
	\STATE If $a_{i, \theta} = 0$ then go to
	step~\ref{step:simplest first}, otherwise output $\theta$.
\end{algorithmic}
\end{algorithm}

Simple calculation shows that
\begin{align*}
	\Prob{\text{Alg.\ returns } \theta}
	= \frac{a_{\theta}}{\sum_{\vartheta = 1}^{W} a_{\vartheta}}
\end{align*}
which is the same as \cref{eqn:general probability}
if $\fn{g}{a} = a$.

\section{Adaptive Stochastic Boosting, a new method}
\label{sec:adaptive stochastic}

The main idea is to combine the quantum ensemble classifier with iterative re-weighting.
	The algorithm iteratively weights the given classifiers and the training samples. The classifiers give  predictions for the training samples, and instead of selecting one weak learner according to some distribution, we take their weighted sum, where the weights reflect their accuracies. Then we weight the samples with respect to how difficult they are for the classifiers to learn. In the next iteration this distribution on the samples is taken into account. After a few iterations, the obtained weighted sum of classifiers is declared to be the  final one, and this is used to classify further samples.

	\Cref{sec:algorithm} describes the details and the convergence is shown using linear algebra.
	The main points of our implementations are sketched in \cref{sec:realizations}.
	The comparison of our implementations with the widely used AdaBoost is in \cref{sec:comparisons}.

\subsection{Algorithm} \label{sec:algorithm}
	As before, there are $N$ training samples, $W$ weak learners or classifiers, and each weak learner gives a prediction to every sample.

	In every iteration, an aggregation of the weights of weak learners is calculated and
the weights of the samples are determined by the predictions, the real labels, and the weights of the weak learners. These weights are always normalized to sum up to one (to form probability distributions). We determine the samples for the next iteration by sampling $N$ (not necessarily distinct) elements according to their weights. The set of the actual samples of the $t$-th iteration is denoted by $\left\{\left(\underline{x}_i^{(t)},y_i^{(t)}\right)\right\}_{i=1}^N$.

	The prediction for unknown input is given by the outputs of the weak learners weighted by their (final) aggregated weights. The algorithm is described in detail in \cref{algo:sampling}.
	The first part of the algorithm is the training phase: based on the training samples, the final weighting of the classifiers is computed.
	The second part is the test phase, where the model is used to classify new inputs.

\begin{algorithm}[h t b p]
	\caption{The sampling version of the algorithm}
	\label{alg:new_alg}
	\textbf{Training algorithm:}
	\begin{algorithmic}[1]
		\REQUIRE  training set: $(\underline{x}_1,y_1),\dots,(\underline{x}_N,y_N) \in (X, \cbr{0,1})$, \\
			weak learners: $h_1, h_2, \dots, h_W: X \rightarrow \cbr{0,1}$, \\
			number of iterations: $T\in \mathbb{Z}^+$.
		\ENSURE aggregated weight vector of the weak learners:
		$\underline{w}$.
		\STATE \textbf{Initialization:} elements of initial sample set: $\left(\underline{x}_i^{(1)},y_i^{(1)}\right)=(\underline{x}_i,y_i)$ for $i\in\{1,\dots,N\}$,\\
		aggregated weight vector of  weak learners: $\underline{w}_{aggr}^{(1)}=\underline{0}$.
		\FORALL{$t=1,2,\dots,T$}
			\STATE Compute predictions of  weak learners
		    $h_{\theta}(\underline{x}_i^{(t)})$.
			\STATE Calculate accuracy of  predictions $\displaystyle a_{\theta}^{(t)}=\sum_{i=1}^N 1-\left|h_{\theta}\left(\underline{x}_i^{(t)}\right)-y_i^{(t)}\right|$.
			\STATE Calculate weights of weak learners (normalization)
			$w_{\theta}^{(t)}=a_{\theta}^{(t)}/\norm{\underline{a}^{(t)}}_{1} $.
			\STATE Update aggregated weights
			$\underline{w}_{aggr}^{(t+1)}=(\underline{w}_{aggr}^{(t)}+\underline{w}^{(t)})/\norm{\underline{w}_{aggr}^{(t)}+\underline{w}^{(t)}}_{1}$.
			\STATE Calculate error on sample points
			$ q_i^{(t+1)} = \sum_{\theta=1}^{W}w_{\theta}^{(t)}\left|h_{\theta}\left(\underline{x}_i^{(t)}\right)-y_i^{(t)}\right|$.
			\STATE Calculate weights of samples (normalization)
			$p_i^{(t+1)}= q_i^{(t+1)}/\norm{\underline{q}^{(t+1)}}_1$.
			\STATE Sample $N$ elements from $\left\{\left(\underline{x}_i^{(t)},y_i^{(t)}\right)\right\}_{i=1}^N$ according to distribution $\underline{p}^{(t+1)}$ to get the next sample set $\left\{\left(\underline{x}_i^{(t+1)},y_i^{(t+1)}\right)\right\}_{i=1}^N$.
		\ENDFOR
		\STATE $\underline{w} = \underline{w}_{aggr}^{(T+1)}$.
	\end{algorithmic}
		\vspace*{3ex}
		\textbf{Classification algorithm:}
	\begin{algorithmic}[1]
		\REQUIRE sample $\underline{x} \in \mathbb{R}^d$ with unknown label,
		weak learners $h_{\theta},\ \theta\in \{1,\dots,W\}$,
		weight vector of learners $\underline{w} \in \mathbb{R}^W$.
		\ENSURE final prediction: $h_f(\underline{x})$.
		\IF{$\sum_{\theta=1}^{W}h_{\theta}(\underline{x}) \cdot w_{\theta}>0.5$}
			\STATE $h_f(\underline{x})=1$
		\ELSE
			\STATE $h_f(\underline{x})=0$
		\ENDIF
	\end{algorithmic}
	\label{algo:sampling}
\end{algorithm}

The algorithm can be formalized with matrices.
Let $M \in \mathbb{R}^{N\times W}$
be the matrix that describes the errors, i.e., $M_{i,j}=1-a_{i,j}$ which is $0$ if the $j$-th classifier classifies the $i$-th sample correctly,
otherwise $M_{i,j}=1$.
The matrix $M'\in \mathbb{R}^{W\times N}$ is obtained from $M$ by transposing the matrix and negating its entries,
so it is the matrix of accuracies,  $M'_{i,j}=a_{j,i}$.
Then, at some iteration both the vector $\underline{p}^{(t)}$ formed by the errors on samples
and the vector $\underline{w}^{(t)}$ formed by the accuracies of weak learners can be expressed easily.
Notice that
$\underline{q}^{(t+1)}=M\underline{w}^{(t)}$ and similarly, $\underline{a}^{(t)}=M'\underline{p}^{(t)}$. Moreover, $\underline{p}^{(t+1)}$ and $\underline{w}^{(t+1)}$ are just their normalized versions, where the sum of the coordinates are equal to 1.
Combining these, we obtain that  $\underline{w}^{(t+1)}=\alpha_t \cdot M' M\underline{w}^{(t)}$, so $\underline{w}^{(T)}=\alpha \cdot (M' M)^{T-1}\underline{w}^{(1)}$, where $\underline{w}^{(1)}=\alpha_1 \cdot M'\underline{p}^{(1)}$ and ${p}_i^{(1)}=1/N, ~\forall i\in\{1,\dots,N\}$, the $\alpha$'s are the normalizing factors to ensure that the sum of coordinates is 1.

Notice that the matrix $M'M$ has a nice combinatorial meaning:
its $(i,j)$ entry is the number of samples where $h_i$ is correct but $h_j$ errs.
Unfortunately, this is not necessarily a symmetric matrix but look at $M M'$.
This is an $N \times N$ matrix in which the $(i,j)$ entry contains the number of weak learners that are not correct on $x_i$ but correct on $x_j$.
Now in the case when together with every weak learner we also have its negation, this matrix is symmetrical.
Therefore, any sequence $(M M')^t \underline{v}/\norm{v}_1$ converges to an eigenvector if $M M' \underline{v} \ne 0$.
The rate of convergence is exponential in the gap between the largest and second largest absolute values of eigenvalues, for which the eigensubspaces are not orthogonal to  $\underline{v}$.
Since $(M' M)^t = M' (M M')^{t-1} M$,
this implies the following:
\begin{theorem}
Assuming that with every weak learner its negation is also included, the weight vectors $\underline{w}^T$ of  \cref{alg:new_alg} are convergent,
the rate of convergence depends on the eigenvalue gaps of matrix $M M'$.
\qed
\end{theorem}

\subsection{Experiments} \label{sec:realizations}
	The algorithm was implemented in the Python language, using the Jupyter Notebook environment. We designed different versions based on sampling and on the matrix form to see their behavior in practice.

In general, a classification algorithm for a new input computes a value between $0$ and $1$ and depending on some threshold, the final answer will be either $0$ or $1$.
When one examines the quality of such a result, it clearly depends on the choice of the threshold value. On the other hand, the accuracy  is not a good measure in several applications.

In experiments instead of accuracy we used Area Under Curve (AUC) of the Receiver Operating Characteristics (ROC) \cite{green1966signal}.
A nice property of AUC is its threshold invariance. To define it, we have to present some other definitions first.
Let us consider the label of value 1 to be the positive outcome and 0 to be the negative one.

False Positive Rate (FPR) is the number of samples that were classified as positive
 but were negative in fact
 (i.e., $h(\underline{x}_i) = 1$ but $y_i=0$), divided by the number of all the negative samples.
True Positive Rate (TPR) is the ratio of the number of correctly classified positive samples to the number of all the positive samples.
The ROC (Receiver Operating Characteristic) curve plots FPR and TPR at different classification thresholds.
AUC is the Area Under the ROC Curve.

\begin{figure}[!ht]
	\center
	\begin{tikzpicture}
	\node[anchor=south west,inner sep=0] at (0,0) {\includegraphics[width=60mm]{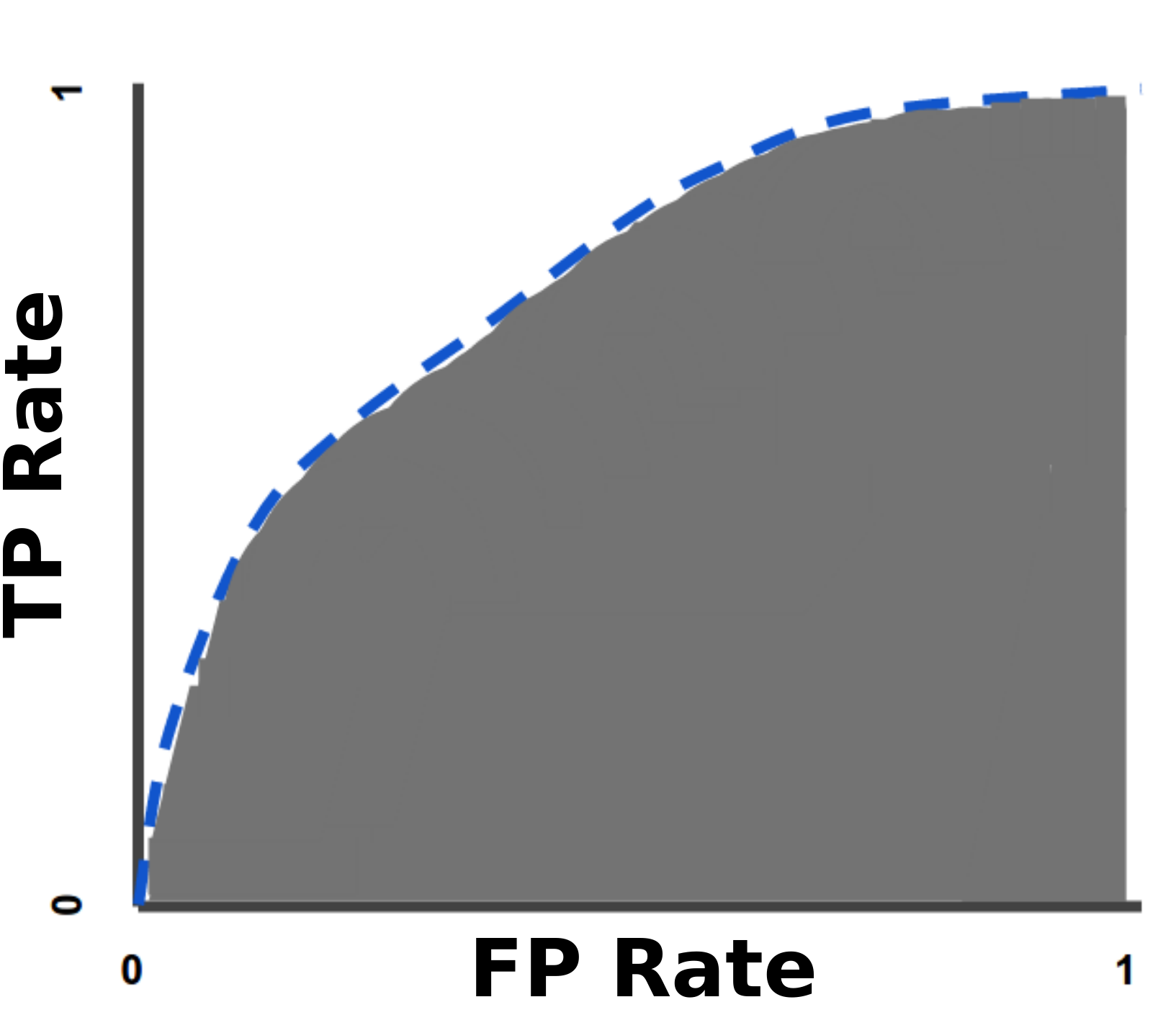}};
	\node[white] (AUC) at (4,2) {\textbf{AUC}};
	\node[blue] (ROC) at (2.4,4) {\textbf{ROC}};
	\fill[white] (2,0) rectangle (4.5,0.45);
	\fill[white] (0,4) rectangle (0.5,1.5);
	\node (FP) at (3.25,0.25) {\textbf{FPR}};
	\node[rotate=90] (TP) at (0.2,2.75) {\textbf{TPR}};
	\end{tikzpicture}
	\caption{An example of ROC curve and AUC.}
	\label{fig:AUC}
\end{figure}
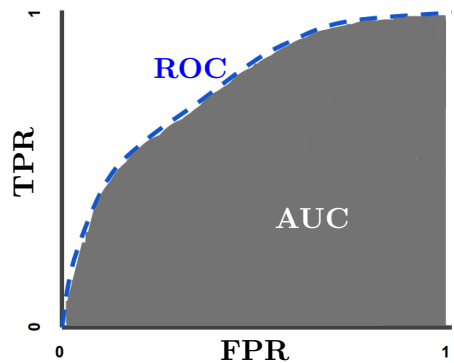

Clearly, AUC is between 0 and 1. An exact classification achieves
$\textrm{AUC}=1$, and
 a classification that assigns 0 or 1, independently of the input has $\textrm{AUC}=1/2$.
Notice that $\textrm{AUC}=0$ is also good in practice, since switching the labels results in $\textrm{AUC}=1$.

\subsubsection{Realization by sampling}
	\label{sec:sampling}
	In our experiments we used decision stumps as weak learners. Three parameters describe a decision stump: the property it refers to (a coordinate of the input vector); the threshold (to separate the 0 and the 1 predictions); and the side of the threshold which contains the 1 predictions.

\subsubsection{Realization with matrix multiplication}

Based on the description after \cref{alg:new_alg} in each step we multiplied the current weight vector by the matrix $M' M$. (In the case when the number of weak learners is much higher than the number of training samples, the multiplication of the appropriate vector by $M M'$ is of course, a better choice.)

In our experiments  the matrix elements were not only $0$ and $1$ but we used elements from the interval $[0, 1]$. Instead of watching only on which side of the threshold of the decision stump the sample is, we look at the distance from the threshold as well. This way the ``certainty'' of the classification is taken into account.

Moreover, in the actual algorithm, making it more similar to the sampling method, we modify the $\underline{p}^{(t)}$ vectors by keeping each coordinate $p_i^{(t)}$ with probability $p_i^{(t)}$, otherwise changing it to $0$.

\subsection{Comparison of the results}
\label{sec:comparisons}
A randomly chosen tenth of the data set was used as validation set, the remaining nine tenths was the training set. We used the AUC value
for measuring the quality of the learning due to AUC's threshold independence.
It means that in fact we do not use a concrete threshold (e.g., 0.5) as written in the previous sections, but only observe if the algorithm gives a good separation of the samples.

We compare five algorithms: AdaBoost from the \emph{ensemble} package of \emph{sklearn}; the expected output of the quantum classifier; and three versions of the new algorithm (realizations using sampling, matrix multiplication, and eigenvector). These algorithms, except for AdaBoost, were implemented by ourselves. In the sampling and matrix multiplication versions the number of iterations was set to $T=10$.

As each data set was randomly split to validation set and training set, we can get different results when we run the tests again. Therefore, we give the minimum, maximum and average results of ten independent executions and ordinary running times.

We observed that the iterative versions of our algorithm (mostly the one we call sampling and the matrix multiplication) tend to overfit if there are only few training samples. That is why we decided to present two kinds of results for these methods: one derived from the AUC values of the 10th iterations of each execution and the other from the maximal AUC values of each execution. The latter may be interpreted as a simulation of early stopping: if overfitting is observed (when the AUC values start to decline), the execution is stopped before the end of the ten iterations.

We tested the methods on multiple data sets. The results are presented for two, publicly available data sets.

\begin{itemize}
	\item Cleveland\footnote{\label{note2}\raggedright\url{https://jamesmccaffrey.wordpress.com/2018/03/14/datasets-for-binary-classification/}}: the task is to recognize heart diseases. Our algorithms generally give better results than AdaBoost. ($N=272$, $d=13$, $W=286$)
	\item Banknote\cref{note2}: the task is to distinguish false banknotes from real ones. All methods give (nearly) perfect results. ($N=1234$, $d=4$, $W=88$)
\end{itemize}
\begin{table}[!ht]
	\centering
	\begin{tabular}{c|ccc|ccc}
		algorithm$\setminus$data set & \multicolumn{3}{c|}{Cleveland} & \multicolumn{3}{c}{Banknote} \\
		& Min & Max & Avg & Min & Max & Avg \\
		\hline
		AdaBoost                & 0.75 & 0.96 & 0.86 & 1.00 & 1.00 & 1.00 \\
		Expectation of quantum  & 0.84 & 0.98 & 0.91 & 0.88 & 0.97 & 0.94 \\
		Sampling (last)         & 0.74 & 0.98 & 0.86 & 0.98 & 1.00 & 0.99 \\
		Sampling (max)          & 0.88 & 0.99 & 0.93 & 0.98 & 1.00 & 0.99 \\
		Matrix multiplication (last)  & 0.83 & 0.99 & 0.91 & 0.89 & 0.97 & 0.94 \\
		Matrix multiplication (max)  & 0.85 & 0.99 & 0.92 & 0.92 & 0.98 & 0.96 \\
		Eigenvector             & 0.85 & 0.98 & 0.91 & 0.90 & 0.98 & 0.95 \\
	\end{tabular}
	\caption{AUC values for the different data sets}
	\label{tab:acc}
\end{table}
\begin{table}[!ht]
	\centering
	\begin{tabular}{c|cc}
		algorithm$\setminus$data set & Cleveland & Banknote\\
		\hline
		AdaBoost                & 0.12 s & 0.17 s\\
		Expectation of quantum  & 0.08 s & 0.09 s\\
		Sampling                & 0.40 s & 0.53 s\\
		Matrix multiplication   & 0.19 s & 0.21 s\\
		Eigenvector             & 0.20 s & 3.67 s\\
	\end{tabular}
	\caption{Running times for the different data sets}
	\label{tab:time}
\end{table}

The results show that among our methods, usually the realization with sampling gives the best classification results, but most of the times it is quite slow (maybe its code could be optimized). We can see that for certain data sets, any of our methods can give better results than AdaBoost, but only the sampling version can match it nearly in every case.

\section{Conclusions}

In this paper we show a linear time and a constant time equivalent classical algorithm to the quantum ensemble classifier by Schuld and Petruccione and suggested a new adaptive boosting algorithm as a combination of ideas taken from the classical equivalent classifier and AdaBoost.
We considered three realizations of the new algorithm and experimented on publicly available data sets. We found that if there are a lot of training samples, the matrix $MM'$ (or if the number of weak learners is large, the matrix $M'M$) is large, and thus the calculation of the eigenvector can be slow depending on which matrix is used in the computation. If there are few training samples, the sampling version may overfit quickly.

{
\bibliographystyle{actaplaindoi}
\bibliography{quantum}
}

\end{document}